\documentclass{article}

\usepackage{amssymb,amsmath,amsthm} 
\usepackage{pstricks,pst-node,pst-plot}
\usepackage{authblk}
\usepackage{enumerate}

\title{Polynomial time recognition of squares of ptolemaic graphs and 3-sun-free split graphs}

\author[1]{Van Bang Le}
\author[2]{Andrea Oversberg}
\author[2]{Oliver Schaudt}

\affil[1]{\small Institut f\"ur Informatik, Universit\"at Rostock, Rostock, Germany\\
\texttt{le@informatik.uni-rostock.de}}
\affil[2]{\small Institut f\"ur Informatik, Universit\"at zu K\"oln, K\"oln, Germany\\ 
\texttt{$\{$oversberg,schaudt$\}$@zpr.uni-koeln.de}}

\begin{document}

\newtheorem{theorem}{Theorem}
\newtheorem{lemma}{Lemma}
\newtheorem{observation}{Observation}
\newtheorem{corollary}{Corollary}
\newtheorem{claim}{Claim}
\newtheorem{conjecture}{Conjecture}

\newcommand{\dist}{\mbox{dist}}

\maketitle

\begin{abstract}
The square of a graph $G$, denoted $G^2$, is obtained from $G$ by putting an edge between two distinct vertices whenever their distance is two.
Then $G$ is called a square root of $G^2$.
Deciding whether a given graph has a square root is known to be NP-complete, even if the root is required to be a chordal graph or even a split graph.

We present a polynomial time algorithm that decides whether a given graph has a ptolemaic square root.
If such a root exists, our algorithm computes one with a minimum number of edges.

In the second part of our paper, we give a characterization of the graphs that admit a 3-sun-free split square root.
This characterization yields a polynomial time algorithm to decide whether a given graph has such a root, and if so, to compute one.\\

\noindent{\textbf{Keywords:}} Square of graphs, square of ptolemaic graphs, squares of split graphs, recognition algorithm.\\

\noindent{\textbf{2010 MSC:}} 05C75, 05C85, 68R05, 68R10.
\end{abstract}

\section{Introduction}
The \emph{square} of a graph $G$ is the graph $G^2$ obtained from $G$ by putting an edge between any two distinct vertices of distance~2.
Then $G$ is called the \emph{square root} of $G^2$.
While every graph has a square, not every graph admits a square root.
In fact, it is NP-complete to decide whether a given graph has a square root, as was shown by Motwani and Sudan~\cite{MS94}. 
Since then, squares of graphs and square roots have been intensively studied, in both graph theoretic and algorithmic aspects. See, for example, \cite{AA10,AA11,CCGKP13,FK12,FLLT12,Lau06,LT10a,MS13} for recent results on this topic.   

One successful approach to deal with this hardness is to ask for square roots that belong to a particular graph class.
This might be useful if one is interested in structural properties of the root graph, such as chordality, bipartiteness or girth conditions.
The negative results in this direction tell us that it is NP-complete to determine whether a graph has a square root that is either chordal~\cite{CL04}, split~\cite{CL04}, or of girth four~\cite{FLLT12}, and, recently announced in \cite{FK12}, of girth five. On the upside, there are polynomial time algorithms for computing a square root that is either a tree (see for example~\cite{CKL,Lau06,LT10b}), a bipartite graph~\cite{Lau06}, a proper interval graph~\cite{CL04}, a block graph~\cite{LT10b}, a strongly chordal split graph~\cite{LT11}, or a graph of girth at least six~\cite{FLLT12}.

\medskip
\noindent
\textbf{Ptolemaic square roots.} Note the contrast between the linear time algorithm for finding a block square root and the NP-hardness of finding a chordal square root.
It seems enticing to investigate what happens \emph{in between} these two classes.
Indeed, two reasonable intermediate graph classes are ptolemaic graphs and strongly chordal graphs.
In this paper, we solve the square root problem for ptolemaic graphs by proving the following main result.

\begin{theorem}\label{result}
It can be decided in $\mathcal O(n^4)$ time whether a given $n$-vertex graph has a ptolemaic square root.
If such a root exists, a ptolemaic square root with a minimum number of edges can be constructed in the same time.
\end{theorem}

A long-standing problem in the research on graph powers is characterizing and recognizing powers of distance-hereditary graphs. This problem stems from a paper by Bandelt, Henkmann and Nicolai~\cite{BHN95}, who where the first to study powers of distance-hereditary graphs. They where, however, not able to give a full characterization of squares of distance-hereditary graphs and this problem remains unsolved until today.

We see our result as an important step towards the solution of the above mentioned problem. Not only is the class of ptolemaic graphs by far the largest subclass of distance-hereditary graphs for which the square root problem is solved. Previous results on subclasses of distance-hereditary graphs, namely the polynomial time algorithms for the recognition of squares of trees and squares of block graphs~\cite{LT10b} can also be subsumed under our result in the following sense.
An implicit feature of our algorithm is that if the input graph admits a square root that is a block graph or a tree, such a root is indeed computed.
(However, the best known algorithm to compute square roots in these two graph classes runs in linear time and is considerably simpler~\cite{LT10b}.)

The optimization aspect of our result is motivated by recent work of Cochefert, Couturier, Golovach, Kratsch and Paulusma~\cite{CCGKP13}.
They introduce the problem of minimizing or maximizing the number of edges in a square root.
Among other results, they give a polynomial time algorithm to compute a square root with a minimum number of edges in the class of graphs of maximum degree 6.

\medskip
\noindent
\textbf{3-sun-free split square roots.} 
As mentioned above, it is NP-complete to decide whether a given graph is the square of some split graph (\cite{CL04}).

A polynomially solvable case in computing split square roots reads as follows. 
A \emph{strongly chordal} graph is a chordal graph that does not contain any 
$\ell$-sun as an induced  subgraph; here an \emph{$\ell$-sun}, $\ell\ge 3$, consists of a stable set $\{u_1, u_2, \ldots, u_\ell\}$ and a clique $\{v_1, v_2, \ldots, v_\ell\}$ 
such that for $i \in \{1,\ldots, \ell\}$, $u_i$ is adjacent to exactly $v_i$ and 
$v_{i+1}$ (index arithmetic modulo $\ell$). 
There is a structural characterization of squares of strongly chordal split graphs, which leads to a quadratic time recognition algorithm (\cite{LT10a,LT11}).

Our second result, Theorem~\ref{thm:recogsquareofs3freesplit} below, extends the polynomially solvable case of strongly chordal split square roots to $3$-sun-free split square roots.
This leaves a larger degree of freedom for the square root, pushing our knowledge on polynomially solvable cases further towards the NP-complete case of general split square roots.

\begin{theorem}\label{thm:recogsquareofs3freesplit}
It can be decided in $\mathcal O(nm^2)$ time whether a given $n$-vertex $m$-edge graph has a $3$-sun-free split square root, 
and if so, such a square root can be constructed in the same time.
\end{theorem}

\medskip
Our paper is structured as follows.
In Section~\ref{basic} we collect relevant notations, definitions, and basic facts.
We also give some more background on ptolemaic graphs. 

We prove our main result, Theorem~\ref{result}, in Section~\ref{ptolemaic}.  
The first of the two ingredients of our algorithm is discussed in Section~\ref{forcededges}.
We show how the structure of the maximal cliques of the input graph already determines an essential part of any ptolemaic square root.
The second ingredient we present in Section~\ref{centers}.
We show that for every maximal clique in the square graph, there is some vertex in any ptolemaic square root whose neighborhood spans this clique.
The results of the last section enable us to determine these vertices.
The whole algorithm is put together in Section~\ref{algorithm}, where we also prove its correctness. 

In Section~\ref{split} we first give a structural characterization of squares of $3$-sun-free split graphs in terms of four forbidden induced subgraphs and of the maximal clique structure. 
Then, Theorem~\ref{thm:recogsquareofs3freesplit} will be derived from this characterization.  

We close the paper with a short discussion of our results and propose some questions for further research in Section~\ref{discussion}.

\section{Basic facts and definitions}\label{basic}

All considered graphs are finite and simple.
Let $G$ be a graph and $v \in V(G)$.
By $N_G(v)$ we denote the set of neighbors of $v$ in $G$.
The \emph{closed neighborhood} of $v$ in $G$, that is $N_G(v) \cup \{v\}$, we denote by $N_G[v]$. 
A \emph{clique}, respectively, an \emph{independent set}, in $G$ is a set of pairwise adjacent, respectively, non-adjacent vertices, in $G$. 
For a subset $X \subseteq V(G)$, we denote by $G[X]$ the subgraph induced by $X$.
If two graphs $G$ and $H$ are isomorphic, we may simply write $G \cong H$.

Let $u,v \in V(G)$.
The distance of $u$ and $v$ in $G$ we denote by $\dist_G(u,v)$.
For any $k \ge 1$, $G^k$ denotes the \emph{$k$-th power} of $G$.
That is the graph on $V(G)$ where any two distinct vertices are adjacent if and only if their distance in $G$ is at most $k$.
$G^2$ is called the \emph{square} of $G$ and $G$ is called the \emph{square root} of $G^2$.

Let $u,v$ be two non-adjacent vertices of $G$. 
A subset $S \subseteq V(G)$ is a \emph{$(u,v)$-separator} if $u$ and $v$ belong to different connected components of $G-S$. 
A \emph{separator} is a $(u,v)$-separator for some non-adjacent vertices $u,v \in V(G)$.
We speak of a \emph{minimal separator} if it is not properly contained in another $(u,v)$-separator. 
A \emph{minimal clique separator} is a minimal separator that is a clique.

\medskip
A graph $G$ is called \emph{distance-hereditary} if for all vertices $u,v \in V(G)$ any induced path between $u$ and $v$ is a shortest path.
Distance-hereditary graphs were introduced by Bandelt and Mulder~\cite{BM86}.
It is well-known that this class can be recognized in linear time~\cite{HM90}.

Powers of distance-hereditary graphs have been studied by Bandelt, Henkmann and Nicolai~\cite{BHN95}.
An important subclass of distance-hereditary graphs are the so-called ptolemaic graphs.
A connected graph $G$ is called \emph{ptolemaic} if for every four vertices $u,v,w,x$ the \emph{ptolemaic inequality} holds: \[\dist_G(u,v) \dist_G(w,x) \le \dist_G(u,w) \dist_G(v,x) + \dist_G(u,x) \dist_G(v,w).\] 
We need the following characterization of ptolemaic graphs.
For any graph $H$ we say that $G$ is \emph{$H$-free} if $G$ does not contain an induced subgraph that is isomorphic to $H$. For a positive integer $\ell$, let $P_\ell$ denote the path on $
\ell$ vertices and $\ell-1$ edges, and $C_\ell$ the cycle on $\ell$ vertices and $\ell$ edges. 
A \emph{gem} is the graph displayed in Fig.~\ref{gem-fig}.
A graph is \emph{chordal} if it is $C_\ell$-free for all $\ell\ge 4$. 

\begin{theorem}[Howorka \cite{How81}]\label{characterization-ptolemaic}
For every graph $G$, the following statements are equivalent.
\begin{enumerate}[\em (i)]
	\item $G$ is ptolemaic;
	\item $G$ is gem-free chordal;
	\item $G$ is $C_4$-free distance-hereditary;
	\item for all vertices $u,v \in V(G)$ of distance two, $N_G(u) \cap N_G(v)$ is a minimal clique $u,v$-separator.
\end{enumerate}
\end{theorem}

\begin{figure}[htb]
\begin{center}
\psset{unit=1.5cm}
\begin{pspicture}(0,0)(3,1.2)

\cnode(0,0.5){0.1cm}{v_1}
\cnode(1,0.2){0.1cm}{v_2}
\cnode(1.5,1){0.1cm}{v_3}
\cnode(2,0.2){0.1cm}{v_4}
\cnode(3,0.5){0.1cm}{v_5}

\ncarc[arcangle=0]{-}{v_1}{v_2}
\ncarc[arcangle=0]{-}{v_2}{v_3}
\ncarc[arcangle=0]{-}{v_3}{v_4}
\ncarc[arcangle=0]{-}{v_4}{v_5}
\ncarc[arcangle=0]{-}{v_3}{v_1}
\ncarc[arcangle=0]{-}{v_3}{v_5}
\ncarc[arcangle=0]{-}{v_2}{v_4}

\end{pspicture}
\end{center}
\caption{the gem}
\label{gem-fig}
\end{figure}

It follows that a ptolemaic graph is always gem-free chordal.
In our proofs we make extensive use of this particular fact without explicitely refering to the above theorem.

A \emph{split graph} is a graph whose vertex set can be partitioned into a clique and an independent set. It is well known that split graphs are exactly the chordal graphs without induced $2K_2$ (the complement of the $4$-cycle $C_4$). 
For more information on graph classes, their definitions and properties we refer to the book by Brandst\"adt, Le and Spinrad~\cite{BLS99}.

\section{Ptolemaic square roots}\label{ptolemaic}

We make use of the following property of squares of ptolemaic graphs later.
  
\begin{theorem}[\cite{BHN95,DD1987,Lub1987,Raych1992}]\label{square-of-ptolemaic-is-chordal}
Squares of ptolemaic graphs are chordal.
\end{theorem}

It is known that ptolemaic graphs are strongly chordal, and squares of strongly chordal graphs are strongly chordal as well (see \cite{DD1987,Lub1987,Raych1992}).

\subsection{Forced edges in a ptolemaic square root}\label{forcededges}

In this section we show how the structure of the maximal cliques of the square of a ptolemaic graph determines an essential part of any ptolemaic square root.
For this, we need the following concept.

Let us say that a \emph{pseudo-$P_5$} in a graph $H$ is an ordered 5-tupel of distinct vertices $(v_1,v_2,v_3,v_4,v_5)$ such that 
\begin{enumerate}[(i)]
	\item $v_2v_3,v_3v_4 \in E(H)$,
	\item $\mbox{dist}_H(v_1,v_2),\mbox{dist}_H(v_4,v_5) \le 2$,
	\item $\mbox{dist}_H(v_1,v_3) = \mbox{dist}_H(v_2,v_4) = \mbox{dist}_H(v_3,v_5) = 2$,
	\item and $\mbox{dist}_H(v_1,v_4),\mbox{dist}_H(v_1,v_5),\mbox{dist}_H(v_2,v_5) \ge 3$.
\end{enumerate}

In particular, an induced $P_5$ is a pseudo-$P_5$.
Fig.~\ref{Pseudo-P5} shows another way of how a pseudo-$P_5$ may appear.

\begin{figure}[htb]
\begin{center}
\psset{unit=1.5cm}
\begin{pspicture}(0,0)(4,1.5)

\cnode(0,1){0.1cm}{v_1}
\cnode(1,1){0.1cm}{v_2}
\cnode(2,1){0.1cm}{v_3}
\cnode(3,1){0.1cm}{v_4}
\cnode(4,1){0.1cm}{v_5}
\cnode(1,0){0.1cm}{x}
\cnode(3,0){0.1cm}{x'}

\ncarc[arcangle=0]{-}{v_2}{v_3}
\ncarc[arcangle=0]{-}{v_3}{v_4}
\ncarc[arcangle=0]{-}{x}{v_1}
\ncarc[arcangle=0]{-}{x}{v_2}
\ncarc[arcangle=0]{-}{x}{v_3}
\ncarc[arcangle=0]{-}{x'}{v_5}
\ncarc[arcangle=0]{-}{x'}{v_3}
\ncarc[arcangle=0]{-}{x'}{v_4}

\nput{90}{v_1}{$v_1$}
\nput{90}{v_2}{$v_2$}
\nput{90}{v_3}{$v_3$}
\nput{90}{v_4}{$v_4$}
\nput{90}{v_5}{$v_5$}

\end{pspicture}
\end{center}
\caption{$(v_1,v_2,v_3,v_4,v_5)$ form a pseudo-$P_5$.}
\label{Pseudo-P5}
\end{figure}
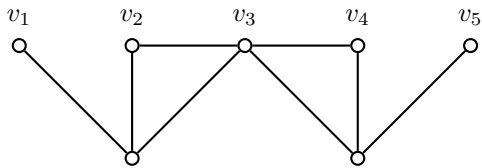

Note that the set $\{v_1,v_2,v_3,v_4,v_5\}$ induces a gem in $H^2$.
In this gem, the sequence $(v_1,v_2,v_4,v_5)$ is an induced $P_4$.
As the next lemma shows, the converse of this statement holds if $H$ is ptolemaic.

\begin{lemma}\label{lem:gem}
Let $H$ be a ptolemaic graph and let $G=H^2$.
If a vertex subset $\{v_1,v_2,v_3,v_4,v_5\}$ induces a gem in $G$ where $(v_1,v_2,v_4,v_5)$ is the induced $P_4$ of this gem, then $(v_1,v_2,v_3,v_4,v_5)$ is a pseudo-$P_5$ in $H$. 
\end{lemma}

\begin{proof}
We make use of the following observation.

\begin{claim}\label{triangle}
Let $x,y,z \in V(G)$ form a triangle in $G$.
Then there is a vertex $u \in V(H)$ such that $u \in N_H[x] \cap N_H[y] \cap N_H[z]$. 
\end{claim}

If $x,y,z$ form a triangle in $H$, the claim is immediate.
So we may assume that $xy \notin E(H)$.
By Theorem~\ref{characterization-ptolemaic}, $N_H(x) \cap N_H(y)$ is a minimal separator of $x$ and $y$.
Since $\dist_H(x,z), \dist_H(y,z) \le 2$, either $z \in N_H(x) \cap N_H(y)$, or $z$ has a neighbor among $N_H(x) \cap N_H(y)$.
In both cases we found some $u \in V(H)$ such that $u \in N_H[x] \cap N_H[y] \cap N_H[z]$.
This proves Claim~\ref{triangle}.\\

We now prove that $(v_1,v_2,v_3,v_4,v_5)$ is a pseudo-$P_5$ in $H$. 
First suppose that $v_2v_4 \in E(H)$.
As $v_1v_4, v_2v_5 \notin E(G)$, $\dist_H(v_1,v_2) = 2 = \dist_H(v_4,v_5)$.
So let $x_{12},x_{45} \in V(H)$ be such that $v_1x_{12},v_2x_{12} \in E(H)$ and $v_4x_{45},v_5x_{45} \in E(H)$.
As $v_1v_5 \notin E(G)$, $x_{12} \neq x_{45}$.
Moreover, $v_1v_5,v_1v_4,v_2v_5 \notin E(G)$ implies $x_{12} \notin N_H[v_4] \cup N_H[v_5]$ and $x_{45} \notin N_H[v_1] \cup N_H[v_2]$.
By chordality of $H$, $x_{12}x_{45} \notin E(H)$, since otherwise $H[\{v_2,v_4,x_{12},x_{45}\}] \cong C_4$.
Hence, $H[v_1,x_{12},v_2,v_4,x_{45},v_5] \cong P_6$.
As $H$ is distance-hereditary, we obtain $\dist_H(v_1, v_5) = 5$.
But $v_1v_3,v_3v_5 \in E(G)$ and so $\dist_H(v_1,v_5) \le 4$, a contradiction.
Hence, $v_2v_4 \notin E(H)$ and so $\dist_H(v_2,v_4) = 2$.

Next suppose that $v_1v_3 \in E(H)$.
Then $v_1v_4,v_1v_5 \notin E(G)$ implies $v_3v_4,v_3v_5 \notin E(H)$.
By Claim~\ref{triangle} there is a vertex $x_{345}$ such that $x_{345} \in N_H[v_3] \cap N_H[v_4] \cap N_H[v_5]$.
Note that $v_3v_4,v_3v_5 \notin E(H)$, yielding $x_{345} \notin \{v_3,v_4,v_5\}$.
Moreover, $v_2x_{345} \notin E(H)$, since $v_2v_5 \notin E(G)$.

By Claim~\ref{triangle} again, there is a vertex $x_{234}$ with $x_{234} \in N_H[v_2] \cap N_H[v_3] \cap N_H[v_4]$.
As $v_2x_{345} \notin E(H)$, $x_{234} \neq x_{345}$; and $v_2v_4,v_3v_4 \notin E(H)$ implies $x_{234} \notin \{v_3,v_4,v_5\}$.
Since $H$ is chordal, $x_{234}$ and $x_{345}$ must be adjacent in $H$, as otherwise $H[\{v_3, v_4, x_{234}, x_{345}\}] \cong C_4$.

If $v_2v_3 \in E(H)$, $H[\{v_2,v_3,v_4,x_{234},x_{345}\}]\cong \mbox{gem}$, a contradiction.
So $v_2v_3 \notin E(H)$.
By Claim~\ref{triangle}, there is a vertex $x_{123}$ such that $x_{123} \in N_H[v_1] \cap N_H[v_2] \cap N_H[v_3]$.
As $v_2v_3 \notin E(H)$, $x_{123} \notin\{v_2,v_3\}$.
Note that possibly $x_{123}=v_1$.
As $v_1v_4 \notin E(G)$, $x_{234} \notin N_H[v_1]$, and so $x_{123} \neq x_{234}$.

If $x_{123}x_{234} \notin E(H)$, $H[\{v_2,v_3,x_{123},x_{234}\}]\cong C_4$, a contradiction.
So $x_{123}x_{234}$ $\in~E(H)$, which yields $H[\{v_2,v_3,x_{123},x_{234},x_{345}\}]\cong \mbox{gem}$, another contradiction.
Hence, $v_1v_3 \notin E(H)$.
By symmetry, $v_3v_5 \notin E(H)$, too.
Indeed, $\dist_H(v_1,v_3) = \dist_H(v_3,v_5) = 2$.

It remains to show that $v_2v_3,v_3v_4 \in E(H)$.
Suppose for a contradiction that $v_2v_3 \notin E(H)$.
By Claim~\ref{triangle}, there are vertices $x_{123},x_{234},x_{345}$ such that $x_{123} \in N_H[v_1] \cap N_H[v_2] \cap N_H[v_3]$, $x_{234} \in N_H[v_2] \cap N_H[v_3] \cap N_H[v_4]$, and $x_{345} \in N_H[v_3] \cap N_H[v_4] \cap N_H[v_5]$.

As $v_1v_4,v_2v_5 \notin E(G)$, the vertices $x_{123}$, $x_{234}$, $x_{345}$ are pairwise distinct.
Moreover, $v_1v_3$, $v_1v_4$, $v_2v_3$, $v_3v_4$ $\notin E(H)$ implies $x_{123}, x_{234} \notin \{v_1,v_2,v_3,v_4\}$.
Similarly, $x_{345} \notin \{v_1, v_2, v_3, v_5\}$.
Note that possibly $x_{345} = v_4$.
If $x_{123}x_{234} \notin E(H)$, $H[\{v_2,v_3,x_{123},x_{234}\}]\cong C_4$, a contradiction.
Thus $x_{123}x_{234} \in E(H)$.

Suppose that $v_3v_4 \notin E(H)$.
In particular, $x_{345} \neq v_4$.
Thus $x_{234}x_{345} \in E(H)$, since otherwise $H[\{v_3,v_4,x_{234},x_{345}\}]\cong C_4$, a contradiction to the chordality of $H$.
But then $H[\{v_2,v_3,x_{123},x_{234},x_{345}\}]\cong \mbox{gem}$, another contradiction.
Hence $v_3v_4 \in E(H)$.
But this gives $H[\{v_2,v_3,v_4,x_{123},x_{234}\}]\cong \mbox{gem}$, another contradiction.

Therefore $v_2v_3 \in E(H)$, and, by symmetry, $v_3v_4 \in E(H)$.
This completes the proof.
\end{proof}

A more general version of Lemma~\ref{lem:gem} reads as follows.
Let us say that a \emph{gem-triple} is an ordered triple $(A,B,C)$ of distinct maximal cliques such that 
\begin{enumerate}[(a)]
  \item $A \cap C \neq \emptyset$,
	\item $A \cap C \subseteq B$,
	\item $A \cap B \not\subseteq C$, and
	\item $B \cap C \not\subseteq A$.
\end{enumerate}
See Fig.~\ref{fig:gem-triple} for an illustration.

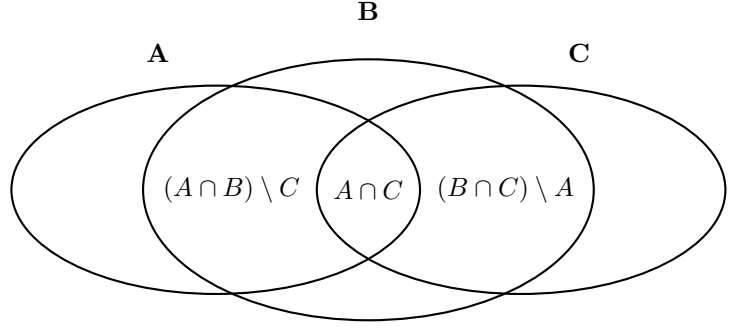
\begin{figure}[h]
\begin{center}
\psset{unit=0.7cm}
\begin{pspicture}(-0.6,-1.2)(10.7,4.4)

\psellipse(2.1,1)(3.9,2)
\psellipse(5,1)(4.3,2.5)
\psellipse(7.9,1)(3.9,2)

\rput(1,3.6){$\mathbf{A}$}
\rput(5,4.4){$\mathbf{B}$}
\rput(9,3.6){$\mathbf{C}$}

\rput(2.4,1){$(A \cap B) \setminus C$}
\rput(5,1){$A \cap C$}
\rput(7.6,1){$(B \cap C) \setminus A$}

\end{pspicture}
\caption{A gem-triple $(A,B,C)$.}
\label{fig:gem-triple}
\end{center}
\end{figure}

\begin{lemma}\label{lem:gem-triple-edges}
Let $H$ be a ptolemaic graph, let $G=H^2$, and let $(A,B,C)$ be a gem-triple in $G$.
Then for all $u \in A \cap C$ and $v \in (A \cup C) \cap B$ with $u \neq v$, $uv \in E(H)$.
\end{lemma}

\begin{proof}
Let $v_2 \in (A \cap B) \setminus C$, $v_3 \in A \cap C$, $v_4 \in (B \cap C) \setminus A$.
Due to the maximality of $A$ and $C$, there are vertices $v_1 \in A \setminus (B \cup C)$ and $v_5 \in C \setminus (A \cup B)$ with $v_1 v_4, v_2 v_5 \notin E(G)$.
Otherwise, $A \cup \{v_4\}$ resp.~$C\cup\{v_2\}$ would be a clique, a contradiction.

By Theorem~\ref{square-of-ptolemaic-is-chordal}, $G$ is chordal. 
Thus, $v_1$ and $v_5$ must be non-adjacent in $G$, as otherwise $G[\{v_1,v_2,v_4,v_5\}] \cong C_4$.
Hence, $G[\{v_1,v_2,v_3,v_4,v_5\}]\cong \mbox{gem}$.
Lemma~\ref{lem:gem} implies that $(v_1,v_2,v_3,v_4,v_5)$ is a pseudo-$P_5$ in $H$.
As $v_2,v_3,v_4$ were arbitrary, every edge between $(A \cap B) \setminus C$, resp.~$(B \cap C) \setminus A$, and $A\cap C$ is present in $H$.

Moreover, if for some  $v_3' \neq v_3$ it holds that both $(v_2,v_3,v_4)$ and $(v_2,v_3',v_4)$ are induced $P_3$ in $H$, the chordality of $H$ implies $v_3v_3' \in E(H)$.
This means that $A \cap C$ is a clique in $H$, completing the proof. 
\end{proof}

\subsection{Centers of maximal cliques}\label{centers}

The last section shows that several edges of any ptolemaic square root of a graph are forced.
However, a non-trivial degree of freedom remains for the choice of the other edges of the square root.
This issue is dealt with in the present section.

Let $H$ be any ptolemaic graph and let $G=H^2$.
Let $C$ be a maximal clique of $G$ and $x \in V(H)$ with $N_H[x] = C$.
We call $x$ a \emph{center} of $C$ (with respect to $H$).
As Lemma~\ref{lem:centersexist} below shows, a center exists for every maximal clique $G$.
Note that for every vertex $v \in V(G)$ being in a maximal clique $C$ of $G$ is equivalent to being identical or adjacent to a center of $C$ in $H$.
In the remainder of this paper, we make use of this fact without explicitly repeating it.

Next we prove a sequence of lemmas in order to prepare our algorithm.

\begin{lemma}\label{lem:centersexist}
Let $H$ be a ptolemaic graph and let $G=H^2$.
Every maximal clique of $G$ has a center in $H$.
\end{lemma}

\begin{proof}
Let $Q$ be a maximal clique in $G$, and let $v\in Q$ such that $N_H[v]\cap Q$ is inclusion-maximal. We are going to show that $N_H[v]=Q$. Note that $N_H[v]$ is a clique in $G$. So, by the maximality of the clique $Q$, it suffices to show that $Q\subseteq N_H[v]$.

Assume, by way of contradiction, there is a vertex $x\in Q-N_H[v]$. Then $\dist_H(v,x)=2$, and, by Theorem~\ref{characterization-ptolemaic}, 
\[
 \mbox{$S = N_H(v)\cap N_H(x)$ is a minimal clique $(v,x)$-separator in $H$.}
\]
Let $A$ be the component of $G-S$ containing $v$. 

\begin{claim}\label{claim:centersexist1}
For every vertex $y\in Q\cap A$, $N_H(y)\cap S=S$.
\end{claim}

If $y\in Q\cap A$, then $\dist_H(x,y)=2$. Hence by Theorem~\ref{characterization-ptolemaic}, $S'=N_H(x)\cap N_H(y)\subseteq S$ is an $(x,y)$-separator in $H$. Therefore, $S'=S$, otherwise there would be a path in $H$ connecting $x$ and $y$ using $v$ and a vertex in $S-S'$. Claim~\ref{claim:centersexist1} follows.

\begin{claim}\label{claim:centersexist2}
For every vertex $s\in S$ and every vertex $q\in Q$, $\dist_H(s,q)\le 2$.
\end{claim}

If $q\in S$, the claim follows from the fact that $S$ is a clique in $H$. If $q\in A$, the claim follows from Claim~\ref{claim:centersexist1}. Let $q\in B$, where $B$ is another component of $G-S$. Since $\dist_H(q,v)=2$, $q$ must have some neighbor in $S$. Since $S$ is a clique, $\dist_H(s,q)\le 2$. Claim~\ref{claim:centersexist2} follows.\\

Consider now a vertex $s\in S$. The maximality of the clique $Q$ in $G$ and Claim~\ref{claim:centersexist2} imply that $s$ must belong to $Q$. On the other hand, by Claim~\ref{claim:centersexist1}, $N_H[v]\cap Q\subseteq N_H[s]\cap Q$, and this inclusion is proper because $sx\in E(H)$ but $vx\not\in E(H)$. This contradicts the choice of $v$. Hence $N_H[v]=Q$ as claimed.
\end{proof}

\begin{lemma}\label{domination}
Let $H$ be a ptolemaic graph and let $G=H^2$.
For every vertex $u \in V(H)$, there is a center $x$ with $N_{H}[u] \subseteq N_{H}[x]$.
\end{lemma}

\begin{proof}
As $N_H[u]$ is a clique in $G$, there must be some maximal clique $C$ of $G$ with $N_H[u] \subseteq C$.
Let $x$ be a center of $C$.
Then $N_{H}[u] \subseteq C = N_{H}[x]$.
\end{proof}

Note that Lemma~\ref{domination} implies that the centers of $G$ are exactly the vertices of $H$ with inclusionwise maximal closed neighborhood.

\begin{lemma}\label{P3center}
Let $H$ be a ptolemaic graph and let $G=H^2$.
Let $x,y \in V(H)$ be two centers with $\dist_{H}(x,y) = 2$.
Then there is a center among $N_{H}(x) \cap N_{H}(y)$.
Moreover, for any center $z \in N_{H}(x) \cap N_{H}(y)$, $(N_H[x],N_H[z],N_H[y])$ is a gem-triple.
\end{lemma}

\begin{proof}
As $\dist_{H}(x,y) = 2$, there is some vertex $v \in N_{H}(x) \cap N_{H}(y)$.
By Lemma~\ref{domination}, there is some center $z$ with $N_H[v] \subseteq N_H[z]$.
In particular, $z \in N_{H}(x) \cap N_{H}(y)$, which proves the first part of our claim.

Choose any center $z \in N_{H}(x) \cap N_{H}(y)$.
As $x,y,z$ are centers, $N_H[x]$, $N_H[y]$, and $N_H[z]$ are maximal cliques of $G$ by definition. 

Note that, by Theorem~\ref{characterization-ptolemaic}, $N_H(x)\cap N_H(y)$ is a clique, hence $N_H[x] \cap N_H[y] \subseteq N_H[z]$. 
As $x \in (N_H[x] \cap N_H[z]) \setminus N_H[y]$ and $y \in (N_H[z] \cap N_H[y]) \setminus N_H[x]$, $(N_H[x],N_H[z],N_H[y])$ is a gem-triple.
\end{proof}

Our next lemma enables a key step of our algorithm.
It allows to determine the centers of the maximal cliques of $G$, up to being adjacent twins in $G$. Here, two vertices are \emph{twins} if they have the same neighbors. 

\begin{lemma}\label{lem:center-vs-triple}
Let $H$ be a ptolemaic graph and let $G=H^2$.
Let $A,C$ be two maximal cliques of $G$ with $A \cap C \neq \emptyset$, and let $v$ be a center of $A$.
Then $v \in A \setminus C$ if and only if there is a maximal clique $B$ such that $(A,B,C)$ is a gem-triple.
\end{lemma}

\begin{proof}
First assume that there is a maximal clique $B$ such that $(A,B,C)$ is a gem-triple.
Then, by definition, there is a vertex $w \in B \cap (C \setminus A)$.
Suppose that $v \in A \cap C$.
By Lemma~\ref{lem:gem-triple-edges}, $vw \in E(H)$, which means that $N_H[v] \not\subseteq A$, a contradiction to the choice of $v$.

Now assume that $v \in A \setminus C$.
Let $w$ be a center of $C$.
If $w \in A \cap C$, $vw \in E(H)$ since $N_H[v] = A$.
But this gives $N_H[w] \not\subseteq C$, a contradiction to the choice of $w$.
So, $w \in C \setminus A$.
By Lemma~\ref{P3center}, there is a maximal clique $B$ such that $(A,B,C)$ is a gem-triple, and this completes the proof.
\end{proof}

\subsection{The Algorithm}\label{algorithm}

We now state our algorithm and then discuss its logic.
Let $G$ be the input graph.

\begin{enumerate}
	\item\label{chordalcheck} Check whether $G$ is chordal. If not, return that $G$ does not have a ptolemaic square root.

	\item Compute the maximal cliques of $G$. Let the set of maximal cliques be denoted $\mathcal C$.

	\item Initialize the empty graph $H$ on the vertex set $V(G)$.
	
	\item\label{gem-triple-computation} Determine the gem-triples among $\mathcal C$.
	
	\item\label{gem-stamps-step} For every gem-triple $(A,B,C)$ in $G$, add the edge $uv$ to $E(H)$, for all $u \in A \cap C$ and $v \in (A \cup C) \cap B$ with $u \neq v$.	
	
	\item\label{centerlocation} Peform the following steps for every $A \in \mathcal C$.				
				\begin{enumerate}[(i)]
					\item Compute the set $\mathcal C_A$ of maximal cliques $C \in \mathcal C$ with $A \cap C \neq \emptyset$.
					\item Compute the set $\mathcal C_A'$ of maximal cliques $C \in \mathcal C$	for which there is a $B \in \mathcal C$ such that $(A,B,C)$ is a gem-triple.
					\item Compute the set $\mathcal C_A'' = \mathcal C_A \setminus \mathcal C_A'$.
					\item Compute the vertex set $X_A = \bigcap \mathcal C_A'' \setminus \bigcup \mathcal C_A'$.
				\end{enumerate}
			
	\item\label{assigningcenters} Assign a vertex $x_C \in X_C$ to every $C \in \mathcal C$ in an injective way. If this is not possible, return that $G$ does not have a ptolemaic square root.

	\item\label{placingcenteredges} For every $C \in \mathcal C$ and $v \in C$: if $v \neq x_C$ and $vx_C \notin E(H)$, add the edge $vx_C$ to $E(H)$.
	
		\item\label{crucialstep} Check whether the graph $H$ is a ptolemaic square root of $G$. If yes, return $H$. If not, return that $G$ does not have a ptolemaic square root.
\end{enumerate}

\noindent
In Step~\ref{gem-stamps-step} the forced egdes are included into the potential square root $H$ according to Lemma~\ref{lem:gem-triple-edges}.
Potential centers for $H$ are determined in Steps~\ref{centerlocation} and~\ref{assigningcenters} according to Lemma~\ref{lem:center-vs-triple}.
Step~\ref{placingcenteredges} implements the neighborhoods of these centers.

\begin{lemma}\label{runtime}
The algorithm can be implemented such that it terminates in $\mathcal O(n^4)$ time when applied to an $n$-vertex graph.
\end{lemma}

\begin{proof}
It is well known that chordal graphs can be recognized in linear time.
Moreover, an $n$-vertex chordal graph has at most $n$ maximal cliques, and these maximal cliques can be computed in linear time (cf.~\cite{Gol04}).

In order to compute all gem-triples of $G$, we first determine the intersection of each pair of maximal cliques of $G$.
As $|\mathcal C| = \mathcal O(n)$, this is done in $\mathcal O (n^3)$ time.
We can now check for all $A,B,C \in \mathcal C$ the four conditions that define a gem-triple in $\mathcal O(n)$ time.
Thus, Step~\ref{gem-triple-computation} can be performed in $\mathcal O (n^4)$ time.

In Step~\ref{gem-stamps-step} for every gem-triple $(A,B,C)$ of $G$ we have to add the edge $uv$ to $E(H)$, for all $u \in A \cap C$ and $v \in (A \cup C) \cap B$ with $u \neq v$.
This can be done in $\mathcal O(n^4)$ time as follows.
For every pair of maximal cliques $A,C \in \mathcal{C}$ we initialize an empty set $X_{A,C}$.
For every gem-triple $(A,B,C)$, we set \[X_{A,C} \leftarrow X_{A,C} \cup (A \cap B) \cup (B \cap C).\]
Then we add all new edges between vertices of $X_{A,C}$ and $A \cap C$ to $H$.
Note that we already computed the sets $A \cap B$ and $B \cap C$ in Step~\ref{gem-triple-computation}.
As there are $\mathcal O(n^3)$ many gem-triples and the update-step can be done in $\mathcal O(n)$ time, we obtain a running time of $\mathcal O(n^4)$ for Step~\ref{gem-stamps-step}.

Since we computed the data needed to perform steps~\ref{centerlocation}.(i) and \ref{centerlocation}.(ii) in the previous steps, both steps can be performed in time $\mathcal O(n^2)$ per maximal clique of $G$.
As $|\mathcal C| = \mathcal O(n)$, step~\ref{centerlocation}.(iii) can be done in $\mathcal O(n)$ time per maximal clique of $G$.
The computation of the set $X_A = \bigcap \mathcal C_A'' \setminus \bigcup \mathcal C_A'$ in Step~\ref{centerlocation}.(iv) can be done in $\mathcal O(n^2)$ time, as $|\mathcal C_A''| , |\mathcal C_A'| = \mathcal O(n)$.
Summing up, Step~\ref{centerlocation} can be done in time $\mathcal O(n^3)$.

We next consider Step~\ref{assigningcenters}.
It is straightforward that for each $C \in \mathcal C$ it holds that $X_C$ is an inclusion-maximal set of adjacent twins.
Thus, for each $C,C' \in \mathcal C$, $X_C$ and $X_{C'}$ are either identical or disjoint.
So, there is an injective assignment as described in Step~\ref{assigningcenters} if and only if for all $C \in \mathcal C$ holds 
\begin{equation}\label{eq:assignment}
|X_C|\ge|\{C' \in \mathcal C : X_{C'} = X_C\}|.
\end{equation}
Moreover, if (\ref{eq:assignment}) holds, such an assignment is immediately found.
In view of (\ref{eq:assignment}), Step~\ref{assigningcenters} can be performed in $\mathcal O(n^2)$ time.

As Step~\ref{placingcenteredges} is straightforward, we proceed to Step~\ref{crucialstep}.
Recall that the square root property can be checked in $\mathcal O(\Delta^2 n)$ for an $n$-vertex graph with maximum degree $\Delta$.
Moreover, ptolemaic graphs can be recognized in linear time, as they are exactly the chordal distance-hereditary graphs an both classes can be recognized in linear time (cf.~\cite{Gol04} resp.~\cite{HM90}).
Thus, Step~\ref{crucialstep} can be performed in $\mathcal O(n^3)$ time.

We obtain an overall runtime of $\mathcal O(n^4)$ for our algorithm.
\end{proof}

Thus, the above algorithm terminates in polynomial time.

\begin{lemma}\label{main-lemma}
If the input graph $G$ has a ptolemaic square root, the algorithm puts out a ptolemaic square root of $G$ that has a minimum number of edges.
\end{lemma}

\begin{proof}
Assume that $G$ has a ptolemaic square root.
Then $G$ is chordal by Theorem~\ref{square-of-ptolemaic-is-chordal}, 
and the algorithm proceeds after Step~\ref{chordalcheck}.
Let $H^*$ be a ptolemaic square root of $G$ with a minimum number of edges.

As the algorithm proceeds after Step~\ref{chordalcheck}, for every maximal clique $C$ the sets $\mathcal C_C'$, $\mathcal C_C''$, and $X_C$ are computed.
By Lemma~\ref{lem:centersexist}, every maximal clique $C \in \mathcal C$ has a center, say $x_C^*$, in $H^*$.
Lemma~\ref{lem:center-vs-triple} implies that $x_C^*$ is not contained in any clique of $\mathcal C_C'$, but in every clique of $\mathcal C_C''$.
Thus, $x_C^* \in X_C$.
In particular, the injective assignment of centers as in Step~\ref{assigningcenters} is possible.
So, the algorithm injectively determines a center $x_C$ for every $C \in \mathcal C$.

Next we show that we may assume $E(H) \subseteq E(H^*)$.
As $N_{H^*}[x_C^*] = N_H[x_C]$ and $|\{x_C^*:C \in \mathcal C\}| = |\mathcal C| = |\{x_C:C \in \mathcal C\}|$,
we can relabel the vertices of $H^*$ such that $x_C^* = x_C$ for every $C \in \mathcal C$.
As mentioned above, for each $C \in \mathcal C$ the set $X_C$ is an inclusion-maximal set of adjacent twins in $G$.
Hence in ${H^*}^2$, the relabeling preserves closed neighborhoods, so the square of $H^*$ does not change.

As $x_C^* = x_C$ for every $C \in \mathcal C$, $\{vx_C : v \in C \setminus \{x_C\}, C \in \mathcal C\} \subseteq E(H^*)$.
Every other edge of $H$ is forced by Lemma~\ref{lem:gem-triple-edges}.
As Lemma~\ref{lem:gem-triple-edges} applies to $H^*$ as well, $E(H) \subseteq E(H^*)$.
In particular, $|E(H)| \le |E(H^*)|$, which is one part of the statement of Lemma~\ref{main-lemma}.

It remains to show that $H$ is a ptolemaic square root of $G$.
From $E(H) \subseteq E(H^*)$ it follows that $E(H^2) \subseteq E({H^*}^2) = E(G)$.
For every $C \in \mathcal C$, $N_H[x_C] = C$ and so $C$ is a clique in $H^2$.
This yields $E(H^2) \supseteq E(G)$, and so $H^2 = G$.

We complete the proof with a sequence of claims showing that $H$ is ptolemaic.
That is, $H$ is chordal and gem-free, by Theorem~\ref{characterization-ptolemaic}.

\begin{claim}\label{centerinclusion}
Let $u,v \in V(H)$ be two distinct vertices that are not centers. 
If for every center $x \in N_H(u)$ holds $x \in N_H(v)$, then $N_H(u) \subseteq N_H[v]$.
\end{claim}

Let $w \in N_H(u)$ be arbitrary.
If $w$ is a center, $w \in N_H[v]$ by assumption.

So assume that $w$ is not a center.
This means the edge $uw$ has been included to $H$ in some Step~\ref{gem-stamps-step}.
Hence, there is a gem-triple $(A,B,C)$ such that $u,w \in A \cap C$, $u \in A \cap C$ and $w \in (A \cap B) \setminus C$, or $w \in A \cap C$ and $u \in (A \cap B) \setminus C$. 
By assumption, $v$ is contained in every maximal clique of $G$ that $u$ is contained in.
Thus $v,w \in A \cap C$, $v \in A \cap C$ and $w \in (A \cap B) \setminus C$, or $w \in A \cap C$ and $v \in (A \cap B) \setminus C$.
This means the edge $vw$ is included to $H$, too.
This proves Claim~\ref{centerinclusion}.

\begin{claim}\label{unknownedges}
Every two vertices $u,v \in V(H)$ with $uv \notin E(H)$ and $uv \in E(H^*)$ satisfy $N_H(u) \subseteq N_H(v)$ or $N_H(v) \subseteq N_H(u)$.
\end{claim}

To prove our claim, let us assume that $N_H(u) \not\subseteq N_H(v)$.
As for every center $x \in V(H)$ holds $N_H(x)=N_{H^*}(x)$, both $u$ and $v$ cannot be centers.
Thus, by Claim~\ref{centerinclusion}, $u$ is adjacent to a center $x$ that is not adjacent to $v$.
Let us show that in $H$, $u$ is adjacent to every center $x$ with $xv \in E(H)$.
Claim~\ref{centerinclusion} then completes the proof.

For this, suppose that there is a center $y \in N_H(v) \setminus N_H(u)$.
We have $xy \notin E(H)$, as otherwise $xy \in E(H^*)$ and so $H^*[\{u,v,x,y\}] \cong C_4$, a contradiction.

By Lemma~\ref{domination}, there is a center $z$ with $N_{H^*}[u] \subseteq N_{H^*}[z]$.
Hence $xz,uz,vz \in E(H^*)$.
As $H^*$ is gem-free, $yz \notin E(H^*)$.
Similarly, there is a center $z'$ with $N_{H^*}[v] \subseteq N_{H^*}[z']$, that is, $uz',vz',yz' \in E(H^*)$.
Since $H^*$ is gem-free, $zz' \in E(H^*), xz' \notin E(H^*)$.
Let $A,B,C$ be the maximal cliques of $G$ corresponding to $x,z,z'$, in that order.
By Lemma~\ref{P3center}, $(A,B,C)$ is a gem-triple of $G$ with $u \in A \cap B \cap C$ and $v \in B \cap C$. 
But then the edge $uv$ has been added to $H$ in Step~\ref{gem-stamps-step} of the algorithm, a contradiction.
This proves Claim~\ref{unknownedges}.

\begin{claim}\label{chordalclaim}
$H$ is chordal.
\end{claim}

First suppose that $H$ contains an induced cycle $C$ of length at least five.
As $H^*$ is chordal, $C$ contains at least two vertices $u,v$ with $uv \notin E(H)$ and $uv \in E(H^*)$.
By Claim~\ref{unknownedges}, $N_H(u) \subseteq N_H(v)$ or $N_H(v) \subseteq N_H(u)$.
But this is a contradiction to the fact that $C$ is an induced cycle of length at least five.

So suppose that $H$ contains an induced $C_4$.
Let $v_1,v_2,v_3,v_4$ be a consecutive ordering of that $C_4$.
As $H^*$ is chordal, we may assume that $v_1v_3 \in E(H^*)$.
Thus $N_H[v_1] \neq N_{H^*}[v_1]$ and $N_H[v_3] \neq N_{H^*}[v_3]$, which in turn implies that both $v_1$ and $v_3$ cannot be centers.

Suppose that $N_H(v_1)$ and $N_H(v_3)$ are incomparable. 
Then Claim~\ref{centerinclusion} implies that there are centers $x \in N_H(v_1) \setminus N_H(v_3)$ and $y \in N_H(v_3) \setminus N_H(v_1)$.
By chordality of $H^*$, $xy \notin E(H)$: otherwise, $H^*[\{v_1,v_3,x,y\}] \cong C_4$.
By Lemma~\ref{domination}, there are centers $x',y'$ with $N_{H^*}[v_1] \subseteq N_{H^*}[x']$ and $N_{H^*}[v_3] \subseteq N_{H^*}[y']$.
As $H^*$ is ptolemaic, $xy' \notin E(H^*)$, since otherwise $H^*[\{x,v_1,v_3,y,y'\}] \cong \mbox{gem}$.
Hence $(N_{H^*}[x],N_{H^*}[x'],N_{H^*}[y'])$ is a gem-triple by Lemma~\ref{P3center}, with $v_1 \in N_{H^*}[x] \cap N_{H^*}[y']$ and $v_3 \in N_{H^*}[x'] \cap N_{H^*}[y']$.
This is a contradiction to the fact that $v_1v_3 \notin E(H)$.
So we may assume that $N_H(v_1) \subseteq N_H(v_3)$.

Suppose that both $v_2$ and $v_4$ are centers.
As $\dist_H(v_2,v_4)= \dist_{H^*}(v_2,v_4) = 2$, Lemma~\ref{P3center} applied to $H^*$ implies the existence of a maximal clique $C$ of $G$ such that $(N_{H^*}[v_2],C,N_{H^*}[v_4])$ is a gem-triple of $G$.
But $v_1,v_3 \in N_{H^*}[v_2] \cap N_{H^*}[v_4]$ and so in Step~\ref{gem-stamps-step}, the edge $v_1v_3$ was included into $H$, a contradiction. 
We may thus assume that $v_4$ is not a center.

As $N_H(v_1)$ and $N_H(v_4)$ are incomparable, there are centers $x \in N_H(v_1) \setminus N_H(v_4)$ and $y \in N_H(v_4) \setminus N_H(v_1)$, by Claim~\ref{centerinclusion}.
By Lemma~\ref{domination}, there are centers $x',y'$ with $N_{H^*}[v_1] \subseteq N_{H^*}[x']$ and $N_{H^*}[v_4] \subseteq N_{H^*}[y']$.
By chordality of $H^*$, $xy \notin E(H)$.
As $H^*$ is ptolemaic, $xy' \notin E(H^*)$, since otherwise $H^*[\{x,v_1,v_4,y,y'\}] \cong \mbox{gem}$.
Hence $(N_{H^*}[x],N_{H^*}[x'],N_{H^*}[y'])$ is a gem-triple by Lemma~\ref{P3center}, with $v_1 \in N_{H^*}[x] \cap N_{H^*}[y']$.
As $v_3 \in N_{H^*}(v_1) \cap N_{H^*}(v_4)$, $v_3 \in N_{H^*}[x'] \cap N_{H^*}[y']$.
This is a contradiction to the fact that $v_1v_3 \notin E(H)$.

Hence, $H$ does not contain any induced cycle of length at least 4.
This completes the proof of Claim~\ref{chordalclaim}.

\begin{claim}\label{gem-freeclaim}
$H$ is gem-free.
\end{claim}

Suppose that $H$ contains a gem as an induced subgraph, say on the vertices $v_1,v_2,v_3,v_4,v_5$.
W.l.o.g.~$(v_1,v_2,v_4,v_5)$ is the induced $P_4$ of this gem.
By Lemma~\ref{domination}, we may assume that $v_3$ is a center.

As $H^*$ is ptolemaic, $v_1v_4 \in E(H^*)$ or $v_2v_5 \in E(H^*)$.
We may assume that $v_1v_4 \in E(H^*)$.
Since $v_1v_4 \notin E(H)$, both $v_1$ and $v_4$ cannot be centers.
This situation is illustrated in Fig.~\ref{gem}.

\begin{figure}[htb]
\begin{center}
\psset{unit=1.5cm}
\begin{pspicture}(0,0)(3,1.2)

\cnode(0,0.5){0.1cm}{v_1}
\cnode(1,0.2){0.1cm}{v_2}
\cnode(1.5,1){0.1cm}{v_3}
\cnode(2,0.2){0.1cm}{v_4}
\cnode(3,0.5){0.1cm}{v_5}

\ncarc[arcangle=0]{-}{v_1}{v_2}
\ncarc[arcangle=0]{-}{v_2}{v_3}
\ncarc[arcangle=0]{-}{v_3}{v_4}
\ncarc[arcangle=0]{-}{v_4}{v_5}
\ncarc[arcangle=0]{-}{v_3}{v_1}
\ncarc[arcangle=0]{-}{v_3}{v_5}
\ncarc[arcangle=0]{-}{v_2}{v_4}
\ncarc[arcangle=20, linestyle=dashed]{-}{v_1}{v_4}

\nput{-125}{v_1}{$v_1$}
\nput{-35}{v_2}{$v_2$}
\nput{35}{v_3}{$v_3$}
\nput{-35}{v_4}{$v_4$}
\nput{-35}{v_5}{$v_5$}

\end{pspicture}
\end{center}
\caption{The situation in Claim~\ref{gem-freeclaim}. The dashed edge is in $E(H^*) \setminus E(H)$.}
\label{gem}
\end{figure}
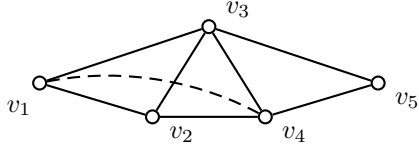

Suppose that $v_2$ is a center.
If $v_5$ is a center, too, then Lemma~\ref{P3center} implies that $(N_{H^*}[v_2],N_{H^*}[v_3],N_{H^*}[v_5])$ is a gem-triple.
Note that $v_1 \in N_{H^*}[v_2] \cap N_{H^*}[v_3]$ and $v_4 \in N_{H^*}[v_2] \cap N_{H^*}[v_3] \cap N_{H^*}[v_5]$.
But this is a contradiction to the assumption $v_1v_4 \notin E(H)$.
So, $v_5$ is not a center.

Suppose that there is a center $x \in N_H(v_4) \setminus N_H[v_2]$.
Then $x \in N_{H^*}(v_4) \setminus N_{H^*}[v_2]$.
By Lemma~\ref{domination}, there is some center $y \in V(H)$ with $N_{H^*}(v_4) \subseteq N_{H^*}[y]$.
Thus $(N_{H^*}[v_2],N_{H^*}[y],N_{H^*}[x])$ is a gem-triple by Lemma~\ref{P3center}, with $v_1 \in N_{H^*}[v_2] \cap N_{H^*}[y]$ and $v_4 \in N_{H^*}[v_2] \cap N_{H^*}[y] \cap N_{H^*}[x]$.
But this contradicts our assumption $v_1v_4 \notin E(H)$.

Hence, for every center $x \in N_H(v_4)$, $x \in N_H[v_2]$.
As both $v_4$ and $v_5$ are not centers and $v_4v_5 \in E(H)$, there is a gem-triple $(A,B,C)$ such that $v_4 \in A \cap B$ and $v_5 \in A \cap B \cap C$ or $v_4 \in A \cap B \cap C$ and $v_5 \in A \cap B$. 
In the first case, $v_2 \in A \cap B$ and in the second $v_2 \in A \cap B \cap C$, but this is a contradiction to $v_2v_5 \notin E(H)$.
So, $v_2$ is not a center.

As $N_H(v_2) \not \subseteq N_H[v_4]$ and $N_H(v_4) \not \subseteq N_H[v_2]$, Claim~\ref{centerinclusion} implies that there are centers $x \in N_H(v_2) \setminus N_H(v_4)$ and $y \in N_H(v_4) \setminus N_H(v_2)$.
Since $N_{H}[x] = N_{H^*}[x]$ and $N_{H}[y] = N_{H^*}[y]$, the chordality of $H^*$ implies that $xy \notin E(H^*)$ and hence $xy \notin E(H)$.
Otherwise, $H^*[\{v_2,v_4,x,y\}] \cong C_4$, a contradiction.
By Lemma~\ref{domination}, there are centers $x',y'$ with $N_{H^*}[v_2] \subseteq N_{H^*}[x']$ and $N_{H^*}[v_4] \subseteq N_{H^*}[y']$.
As $H^*$ is ptolemaic, $xy' \notin E(H^*)$, since otherwise $H^*[\{x,v_2,v_4,y,y'\}] \cong \mbox{gem}$.
Hence $(N_{H^*}[x'],N_{H^*}[y'],N_{H'}[y])$ is a gem-triple by Lemma~\ref{P3center}, with $v_4 \in N_{H'}[x'] \cap N_{H'}[y'] \cap N_{H'}[y]$.
See Fig.~\ref{edge-v1v4} for an illustration.
As $v_1 \in N_{H'}(v_4) \cap N_{H'}(v_2)$, $v_1 \in N_{H'}[x'] \cap N_{H'}[y']$.
But this implies $v_1v_4 \in E(H)$, a contradiction.
This completes the proof of Claim~\ref{gem-freeclaim}.
\end{proof}

\begin{figure}[htb]
\begin{center}
\psset{unit=2cm}
\begin{pspicture}(0,0)(3,2)

\cnode(0,1){0.1cm}{v_1}
\cnode(1,1){0.1cm}{v_2}
\cnode(1.5,2){0.1cm}{v_3}
\cnode(2,1){0.1cm}{v_4}
\cnode(3,1){0.1cm}{v_5}
\cnode(3,0){0.1cm}{y}
\cnode(0,0){0.1cm}{x}
\cnode(2,0){0.1cm}{y'}
\cnode(1,0){0.1cm}{x'}

\ncarc[arcangle=0]{-}{v_1}{v_2}
\ncarc[arcangle=0]{-}{v_2}{v_3}
\ncarc[arcangle=0]{-}{v_3}{v_4}
\ncarc[arcangle=0]{-}{v_4}{v_5}
\ncarc[arcangle=0]{-}{v_3}{v_1}
\ncarc[arcangle=0]{-}{v_3}{v_5}
\ncarc[arcangle=0]{-}{v_2}{v_4}
\ncarc[arcangle=35, linestyle=dashed]{-}{v_1}{v_4}
\ncarc[arcangle=0]{-}{v_1}{y'}
\ncarc[arcangle=0]{-}{v_2}{y'}
\ncarc[arcangle=0]{-}{v_3}{y'}
\ncarc[arcangle=0]{-}{v_4}{y'}
\ncarc[arcangle=0]{-}{v_5}{y'}
\ncarc[arcangle=0]{-}{v_2}{x}
\ncarc[arcangle=0]{-}{y}{v_4}
\ncarc[arcangle=0]{-}{v_1}{x'}
\ncarc[arcangle=0]{-}{v_2}{x'}
\ncarc[arcangle=0]{-}{v_3}{x'}
\ncarc[arcangle=0]{-}{v_4}{x'}
\ncarc[arcangle=0]{-}{y'}{x'}
\ncarc[arcangle=0]{-}{x'}{x}
\ncarc[arcangle=0]{-}{y'}{y}

\nput{-125}{v_1}{$v_1$}
\nput{145}{v_2}{$v_2$}
\nput{35}{v_3}{$v_3$}
\nput{45}{v_4}{$v_4$}
\nput{-35}{v_5}{$v_5$}
\nput{-125}{y}{$y$}
\nput{-35}{x}{$x$}
\nput{-125}{y'}{$y'$}
\nput{-35}{x'}{$x'$}

\end{pspicture}
\end{center}
\caption{The edge $v_1v_4$ is forced by the gem-triple $(N_H[x'],N_H[y'],N_H[y])$.}
\label{edge-v1v4}
\end{figure}
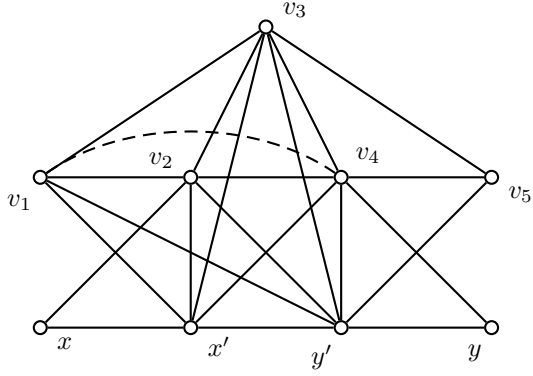

Finally, Theorem~\ref{result} is a direct consequence of the Lemmas~\ref{runtime} and~\ref{main-lemma}.

\section{Squares of 3-sun-free split graphs}\label{split}

In this section we prove Theorem~\ref{thm:recogsquareofs3freesplit}. Recall that deciding if a graph is the square of a strongly chordal split graph can be done in polynomial time (\cite{LT10a,LT11}). This result is based on the following characterization of squares of strongly chordal split graphs; 
the set of all maximal cliques in a graph $G$ is denoted by $\mathcal{C}(G)$. 

\begin{theorem}[\cite{LT10a,LT11}]\label{LT11Char}
$G$ is square of a strongly chordal split graph if and only if $G$ is strongly 
chordal and $\big|\bigcap_{Q \in \mathcal{C}(G)} Q\big|\ge |\mathcal{C}(G)|$.
\end{theorem}

We now are going to extend Theorem \ref{LT11Char} to $3$-sun-free split square roots. 
Our approach is based on the following fact about maximal cliques in squares of $3$-sun-free split graphs. 
A vertex with inclusion-maximal closed neighborhood is called a \emph{maximal} vertex. For split graphs $H=(V(H), E(H))$ we write $H=(C\cup I, E(H))$, meaning $V(H)=C\cup I$ is a partition 
of the vertex set of $H$ into a clique $C$ and an independent set $I$.

\begin{lemma}[\cite{LT10a,LT11}]\label{lem:maxcliqueinsquareofs3-freesplit}
Let $H=(C\cup I, E(H))$ be a connected split graph without induced $3$-sun. Then $Q$ is 
a maximal clique in $H^2$ if and only if $Q=N_H[v]$ for some maximal vertex 
$v\in C$ of $H$.
\end{lemma}

Squares of $3$-sun-free split graphs can be characterized as follows (see Fig.~\ref{fig:H1234} for the graphs $G_1$--$G_4$).

\begin{theorem}\label{thm:squareofs3freesplit}
$G$ is the square of a connected $3$-sun-free split graph if and only if $G$ is $(G_1, G_2$, $G_3$, $G_4)$-free 
and satisfies $\big|\bigcap_{Q \in \mathcal{C}(G)} Q\big|\ge |\mathcal{C}(G)|$.
\end{theorem}
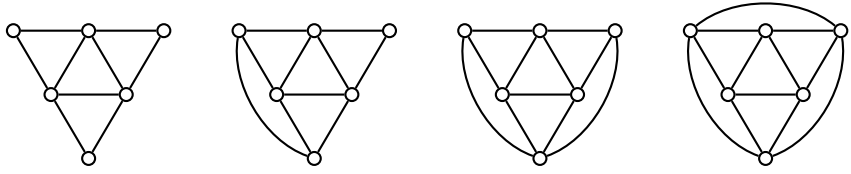
\begin{figure}[ht] 
\begin{center}
\begin{pspicture}(0,0)(11,1.9)

\cnode(0,1.7){0.1cm}{a_1}
\cnode(0.5,0.85){0.1cm}{a_2}
\cnode(1,0){0.1cm}{a_3}
\cnode(1,1.7){0.1cm}{a_4}
\cnode(1.5,0.85){0.1cm}{a_5}
\cnode(2,1.7){0.1cm}{a_6}

\ncarc[arcangle=0]{-}{a_1}{a_2}
\ncarc[arcangle=0]{-}{a_1}{a_4}
\ncarc[arcangle=0]{-}{a_2}{a_4}
\ncarc[arcangle=0]{-}{a_2}{a_3}
\ncarc[arcangle=0]{-}{a_2}{a_5}
\ncarc[arcangle=0]{-}{a_3}{a_5}
\ncarc[arcangle=0]{-}{a_4}{a_5}
\ncarc[arcangle=0]{-}{a_4}{a_6}
\ncarc[arcangle=0]{-}{a_5}{a_6}

\cnode(3,1.7){0.1cm}{b_1}
\cnode(3.5,0.85){0.1cm}{b_2}
\cnode(4,0){0.1cm}{b_3}
\cnode(4,1.7){0.1cm}{b_4}
\cnode(4.5,0.85){0.1cm}{b_5}
\cnode(5,1.7){0.1cm}{b_6}

\ncarc[arcangle=0]{-}{b_1}{b_2}
\ncarc[arcangle=-40]{-}{b_1}{b_3}
\ncarc[arcangle=0]{-}{b_1}{b_4}
\ncarc[arcangle=0]{-}{b_2}{b_4}
\ncarc[arcangle=0]{-}{b_2}{b_3}
\ncarc[arcangle=0]{-}{b_2}{b_5}
\ncarc[arcangle=0]{-}{b_3}{b_5}
\ncarc[arcangle=0]{-}{b_4}{b_5}
\ncarc[arcangle=0]{-}{b_4}{b_6}
\ncarc[arcangle=0]{-}{b_5}{b_6}

\cnode(6,1.7){0.1cm}{c_1}
\cnode(6.5,0.85){0.1cm}{c_2}
\cnode(7,0){0.1cm}{c_3}
\cnode(7,1.7){0.1cm}{c_4}
\cnode(7.5,0.85){0.1cm}{c_5}
\cnode(8,1.7){0.1cm}{c_6}

\ncarc[arcangle=0]{-}{c_1}{c_2}
\ncarc[arcangle=-40]{-}{c_1}{c_3}
\ncarc[arcangle=0]{-}{c_1}{c_4}
\ncarc[arcangle=0]{-}{c_2}{c_4}
\ncarc[arcangle=0]{-}{c_2}{c_3}
\ncarc[arcangle=0]{-}{c_2}{c_5}
\ncarc[arcangle=0]{-}{c_3}{c_5}
\ncarc[arcangle=-40]{-}{c_3}{c_6}
\ncarc[arcangle=0]{-}{c_4}{c_5}
\ncarc[arcangle=0]{-}{c_4}{c_6}
\ncarc[arcangle=0]{-}{c_5}{c_6}

\cnode(9,1.7){0.1cm}{d_1}
\cnode(9.5,0.85){0.1cm}{d_2}
\cnode(10,0){0.1cm}{d_3}
\cnode(10,1.7){0.1cm}{d_4}
\cnode(10.5,0.85){0.1cm}{d_5}
\cnode(11,1.7){0.1cm}{d_6}

\ncarc[arcangle=0]{-}{d_1}{d_2}
\ncarc[arcangle=-40]{-}{d_1}{d_3}
\ncarc[arcangle=0]{-}{d_1}{d_4}
\ncarc[arcangle=40]{-}{d_1}{d_6}
\ncarc[arcangle=0]{-}{d_2}{d_4}
\ncarc[arcangle=0]{-}{d_2}{d_3}
\ncarc[arcangle=0]{-}{d_2}{d_5}
\ncarc[arcangle=0]{-}{d_3}{d_5}
\ncarc[arcangle=-40]{-}{d_3}{d_6}
\ncarc[arcangle=0]{-}{d_4}{d_5}
\ncarc[arcangle=0]{-}{d_4}{d_6}
\ncarc[arcangle=0]{-}{d_5}{d_6}

\end{pspicture}
\end{center}
\caption{$G_1$, $G_2$, $G_3$, and $G_4$.}
\label{fig:H1234}
\end{figure}

\begin{proof}
A \emph{universal vertex} of $G$ is one that is adjacent to every other vertex of $G$. 
Note that, for any connected split-graph $H=(C\cup I, E(H))$, any vertex in $C$ is a universal 
vertex in $H^2$.

Assume that $G=H^2$ for some connected 3-sun-free split graph $H=(C\cup I, E(H))$. 
First, by Lemma~\ref{lem:maxcliqueinsquareofs3-freesplit},
\[ |\mathcal{C}(G)| \le |C|.\]
Furthermore, as $C$ is contained in all maximal cliques in $G$,
\[ |C| \le \big|\bigcap_{Q \in \mathcal{C}(G)} Q\big|.\]
Therefore, 
\[ |\mathcal{C}(G)| \le \big|\bigcap_{Q \in \mathcal{C}(G)} Q\big|.\]

\noindent
Next, let by way of contradiction, $a,b,c,a',b',c'$ be six vertices such that
\[ab, ac, bc, a'b, a'c, b'a, b'c, c'a, c'b\in E(H), aa', bb', cc'\not\in E(H),\]
\noindent
that is, $G[a,b,c,a',b',c']$ is a $G_i$ for some $i=1,2,3,4$.

Let $Q_1, Q_2, Q_3$ be the maximal cliques of $G$ containing $\{a,b,c'\}$, $\{b,c,a'\}$, $\{a,c,b'\}$, respectively. 
By Lemma~\ref{lem:maxcliqueinsquareofs3-freesplit}, $Q_i=N_H[v_i]$ for some (maximal) vertex $v_i\in C$, $i=1,2,3$. In particular, 
\[a,b\in N_H[v_1],\, b,c\in N_H[v_2],\, a,c\in N_H[v_3],\, a\not\in N_H[v_2],\, b\not\in N_H[v_3],\, c\not\in N_H[v_1].\]
By noting that $a,b,c,a',b',c'\in I$ (as none of these vertices is universal in $G$), we conclude that 
$a,b,c,v_1,v_2,v_3$ induce a 3-sun in $H$, a contradiction. 

\medskip\noindent
Now, let $G$ be $(G_1,G_2,G_3,G_4)$-free and satisfy $\big|\bigcap_{Q \in \mathcal{C}(G)} Q\big|\ge |\mathcal{C}(G)|$.

Write $C=\bigcap_{Q \in \mathcal{C}(G)} Q$ and let $\mathcal{C}(G)=\{Q_1,\ldots, Q_q\}$. As $|C|\ge q$, we are able to choose $q$ distinct vertices $c_1, \ldots, c_q$ in $C$. Let $H$ be the split graph with clique $C$, independent set $I=V(G)\setminus C$, and edges $vc_i$ for all $v\in I$ and $1\le i\le q$ with $v\in Q_i$. 

We claim that $G=H^2$.
Indeed, let $xy\in E(G)$. Then there $xy\in Q_i$ for some $i$. If $x\in C$ or $y\in C$, then clearly $xy\in E(H^2)$. If $x,y\in Q_i\setminus C$, then $xc_i, yc_i\in E(H)$, hence $xy\in E(H^2)$. 

Let $xy\in E(H^2)$. If $x\in C$ or $y\in C$, then $xy\in E(G)$ because $C$ is contained in all maximal cliques of $G$. 
So, let $x,y\in I$. Hence there is a vertex $c_i\in C$ with $xc_i, yc_i\in E(H)$. By construction of $H$, $x,y\in Q_i$, showing $xy\in E(G)$. 

We have shown that $G=H^2$, as claimed. It remains to prove that $H$ is 3-sun-free. 
Assume the contrary, and let $v_1=c_i, v_2=c_j, v_3=c_k$, $u_1, u_2, u_3$ induce a 3-sun in $H$. Then, by construction 
of $H$, 
\[
 u_1\in (Q_i\cap Q_j)\setminus Q_k,\, u_2\in (Q_j\cap Q_k)\setminus Q_i,\, u_3\in (Q_i\cap Q_k)\setminus Q_j.
\]
\noindent
Now, by the maximality of the cliques, $u_1$ is non-adjacent to some $x\in Q_k\setminus(Q_i\cup Q_j)$, $u_2$ is non-adjacent to some $y\in Q_i\setminus(Q_j\cup Q_k)$, and $u_3$ is non-adjacent to some $z\in Q_j\setminus(Q_i\cup Q_k)$.
But then $G[u_1,u_2,u_3,x,y,z]$ is one of the $G_1,G_2,G_3,G_4$, a contradiction. This completes the proof.
\end{proof}

We now are going to give an interesting reformulation of Theorem~\ref{thm:squareofs3freesplit}. 
A graph $G$ is said to be \emph{clique-Helly} if $\mathcal{C}(G)$ has the Helly property. 
$G$ is \emph{hereditary clique-Helly} if every induced subgraph of $G$ is  
clique-Helly. (See~\cite{DPS09} for more information on clique-Helly graphs.) 
Prisner~\cite{Prisner93} characterized hereditary clique-Helly graphs as follows.

\begin{theorem}[Prisner \cite{Prisner93}]\label{clique-Helly}
 $G$ is hereditary clique-Helly if and only if $G$ is $(G_1, G_2, G_3, G_4)$-free.
\end{theorem}

It follows that a split graph is hereditary clique-Helly if and only if it is $3$-sun-free. 
With Theorem~\ref{clique-Helly}, Theorem~\ref{thm:squareofs3freesplit} can be reformulated as follows.

\begin{theorem}\label{thm:squareofs3freesplit-v2}
A graph $G$ is the square of a connected hereditary clique-Helly split graph if and only if $G$ is a hereditary clique-Helly graph satisfying $\big|\bigcap_{Q \in \mathcal{C}(G)} Q\big|\ge |\mathcal{C}(G)|$.
\end{theorem} 

We can now give the proof of Theorem~\ref{thm:recogsquareofs3freesplit}.
\begin{proof}[of Theorem~\ref{thm:recogsquareofs3freesplit}]
By Lemma~\ref{lem:maxcliqueinsquareofs3-freesplit}, $G$ has at most $n$ maximal cliques. By \cite{TIAS77}, all maximal cliques in $G$
then can be listed in time $\mathcal{O}(n\cdot m \cdot n)=\mathcal{O}(n^2m)$, and the condition $\big|\bigcap_{Q \in \mathcal{C}(G)} Q\big|\ge |\mathcal{C}(G)|$ can be verified within in the same time. Also, testing if $G$ is hereditary clique-Helly can be done in time $\mathcal{O}(n^2m)$ (see, for instance, \cite{DPS09}). Thus, by Theorem~\ref{thm:squareofs3freesplit-v2}, we can decide in time $\mathcal{O}(n^2m)$ if $G$ is the square of some $3$-sun-free split graph, and if so, the proof of Theorem~\ref{thm:squareofs3freesplit} gives a construction for such a square root $H$ within the same time. 
\end{proof}

\section{Discussion}\label{discussion}

In this paper we have presented a polynomial time algorithm to decide whether a given graph has a ptolemaic square root.
If such a root exists, our algorithm computes a ptolemaic square root with a minimum number of edges.
Let us mention, without a proof, another feature of our algorithm:
if the input graph admits a square root that is a block graph or an acyclic graph, such a root is computed.
However, the best known algorithm to compute square roots in these two graph classes runs in linear time and is considerably simpler~\cite{LT10b}.

Several questions arise now that we can compute ptolemaic square roots.
It is immediate to ask whether ptolemaic $k$-th roots can be efficiently computed.
We did not tackle this question yet, since a more basic question is apparently unanswered: whether one can compute $k$-th roots that are block graphs~\cite{LT10b}.

A question that seems more urgent to us is whether distance-hereditary square roots can be computed efficiently.
Distance-hereditary roots have been considered before in the literature~\cite{BHN95}, yet not from an algorithmic perspective.

Although distance-hereditary graphs share a number of properties with ptolemaic graphs, the two classes behave differently when it comes to graph powers.
To give an example, Fig.~\ref{dhroot} displays the square of a distance-hereditary graph that does not admit a ptolemaic square root.
Indeed, both of our main tools, Lemma~\ref{lem:gem} and Lemma~\ref{lem:centersexist}, fail to hold for distance-hereditary graphs (see again Fig.~\ref{dhroot}).

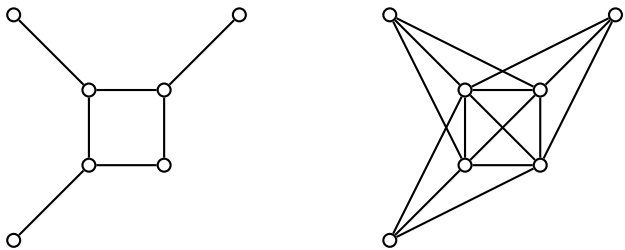
\begin{figure}[htb]
\begin{center}
\psset{unit=1cm}
\begin{pspicture}(0,-1)(7,2.2)

\cnode(0,2){0.1cm}{v_1}
\cnode(1,0){0.1cm}{v_2}
\cnode(1,1){0.1cm}{v_3}
\cnode(2,1){0.1cm}{v_33}
\cnode(2,0){0.1cm}{v_4}
\cnode(3,2){0.1cm}{v_5}
\cnode(0,-1){0.1cm}{v_6}

\ncarc[arcangle=0]{-}{v_1}{v_3}
\ncarc[arcangle=0]{-}{v_2}{v_3}
\ncarc[arcangle=0]{-}{v_2}{v_4}
\ncarc[arcangle=0]{-}{v_2}{v_6}
\ncarc[arcangle=0]{-}{v_3}{v_33}
\ncarc[arcangle=0]{-}{v_33}{v_4}
\ncarc[arcangle=0]{-}{v_33}{v_5}



\cnode(5,2){0.1cm}{v_1}
\cnode(6,0){0.1cm}{v_2}
\cnode(6,1){0.1cm}{v_3}
\cnode(7,1){0.1cm}{v_33}
\cnode(7,0){0.1cm}{v_4}
\cnode(8,2){0.1cm}{v_5}
\cnode(5,-1){0.1cm}{v_6}

\ncarc[arcangle=0]{-}{v_1}{v_3}
\ncarc[arcangle=0]{-}{v_1}{v_2}
\ncarc[arcangle=0]{-}{v_1}{v_33}
\ncarc[arcangle=0]{-}{v_2}{v_3}
\ncarc[arcangle=0]{-}{v_2}{v_33}
\ncarc[arcangle=0]{-}{v_2}{v_4}
\ncarc[arcangle=0]{-}{v_2}{v_6}
\ncarc[arcangle=0]{-}{v_3}{v_33}
\ncarc[arcangle=0]{-}{v_3}{v_4}
\ncarc[arcangle=0]{-}{v_3}{v_5}
\ncarc[arcangle=0]{-}{v_3}{v_6}
\ncarc[arcangle=0]{-}{v_33}{v_4}
\ncarc[arcangle=0]{-}{v_33}{v_5}
\ncarc[arcangle=0]{-}{v_4}{v_5}
\ncarc[arcangle=0]{-}{v_4}{v_6}

\end{pspicture}
\end{center}
\caption{A distance-hereditary graph (left) and its square (right).}
\label{dhroot}
\end{figure}

We also have characterized squares of $3$-sun-free split graphs. Our characterization yields a polynomial time recognition algorithm for such squares. Given the hardness of computing split square roots~(\cite{CL04}) and our polynomial time result (Theorem~\ref{thm:recogsquareofs3freesplit}), it is interesting to ask 
the following question: let $F$ be a fixed split graph, 
what is the computational complexity of computing $F$-free split square roots?  
Our result on 3-sun-free split square roots may pave the way towards such a dichotomy theorem.

\bibliographystyle{amsplain}
\bibliography{graph-powers}

\end{document}